\documentclass[conference]{IEEEtran}    
\usepackage{amsmath, amssymb, amsfonts, amsthm}
\usepackage{bm}
\usepackage{xcolor}
\usepackage{graphicx}
\usepackage{circuitikz}

\usepackage{enumitem}

\newtheorem{lemma}{Lemma}

\title{Optimality of Bang-Bang Switching for Breaking the Chu Limit via Time-Modulated Matching}

\author{\IEEEauthorblockN{Mohamed Akrout$^*$ and Miguel Rodrigo Castellanos$^*$}\\ \vspace{-0.4cm}
\IEEEauthorblockA{\textit{EECS Department}, \textit{University of Tennessee}, Knoxville, TN, USA\\
\{makrout, mrcastellanos\}@utk.edu}}
\vspace{-0.2cm}
\begin{document}

\maketitle
\def\thefootnote{*}\footnotetext{These authors contributed equally to this work}

\begin{abstract}
This paper shows that surpassing the Chu limit on $Q$-factors via time-modulated matching requires non-smooth switching strategies. We derive a nonlinear differential condition to show that differentiable modulation functions are sub-optimal, then show that the optimal switching trajectory for maximizing the violation of the Chu limit is a piecewise-constant (Bang-Bang) profile. Finally, we establish an upper-bound connecting the antenna size, switching speed, and bit error rate, which reveals that switching time becomes longer as antennas become electrically smaller.
\end{abstract}

\section{Introduction}
\subsection{Motivation and prior work}

High-gain wideband antennas are essential for 6G networks. Transmission bandwidths and array sizes have grown steadily over recent decades to meet wireless performance demands. Future devices will operate across millimeter-wave and upper mid-band for high-rate communications while array sizes will continue growing to support MIMO gains. However, fundamental limits on the bandwidth-efficiency product make it difficult to simultaneously increase the gain and bandwidth of a fixed-size antenna \cite{skrivervik2002pcs}. Breaking through this trade-off could significantly enhance future wireless system performance.

Circumventing the assumptions that underlie basic antenna performance bounds allows devices to surpass these fundamental constraints. The Chu limit defines a lower bound on the quality factor $Q$ of electrically small antennas (ESAs) based on their physical volume \cite{chu1948physical}, which, in turn, limits the maximum bandwidth of the device. The derivation of this limit is based on the assumption that the antenna under consideration is a passive linear time-invariant (LTI) system. Prior work has shown that introducing time-varying \cite{LiEtAlBeyondChusLimitWith2019,FrittsEtAlSpaceTimeModulationMultimode2025,mostafa2023antenna} or nonlinear components \cite{ChurchEtAlUhfElectricallySmallBox2014,LoghmanniaManteghiBroadbandParametricImpedanceMatching2022} can improve the bandwidth and efficiency of an antenna. Direct antenna modulation (DAM), in which the antenna excitation is modulated to affect its radiation properties, has also been used for bandwidth broadening \cite{SchabEtAlEnergySynchronousDirectAntenna2020, BroadbooksEtAlNovelContinuousDirectAntenna2025}. The use of these techniques, however, also makes antenna design and analysis more difficult.

We focus on time-varying circuit components as a means to exceed the Chu limit. Previous studies have effectively used time-varying matching networks to increase the bandwidth of ESAs. Time-varying inductors and capacitors have been shown to have the ability to increase the antenna gain-bandwidth product while maintaining circuit stability \cite{LiEtAlBeyondChusLimitWith2019, MekawyEtAlParametricEnhancementRadiationFrom2021, JayathurathnageEtAlTimeVaryingComponentsEnhancing2021}. The design in \cite{FrittsEtAlSpaceTimeModulationMultimode2025} further improves on this concept by using multiple time-varying capacitors with time delays between the modulation function applied to each component. Prior research in this context has largely focused on sinusoidal modulation of the time-varying components. Significant gains may be achieved by optimizing the modulation waveform to maximize different antenna performance metrics.

\subsection{Contributions}
In this paper, we analyze the efficacy of arbitrarily-modulated time-varying inductors in improving ESA $Q$-factors. We first discuss the circuit model of an electrical Chu's antenna with a time-modulated inductor placed at the input and define the time-varying $Q$-factor. We demonstrate the limitations of continuous modulation functions and show that they yield relatively small improvements in terms of the average $Q$-factor. We formulate a functional optimization problem to solve for the modulated inductance that exceeds the Chu bound by a predetermined multiplicative factor. We show that the optimal solution is an on-off modulation function and provide insights into the optimal modulation trajectory. We numerically quantify the trade-off between the $Q$-factor excess $\alpha$ and the modulation time $T_0$. We also discuss the bit error rate (BER) for BPSK DAM over a switching speed range as antennas become electrically smaller.

\section{System model}

We consider an ESA that fits within a sphere of radius $\rho$ and only radiates the $\text{TM}_{01}$ spherical mode. A Chu antenna can be modeled as an RLC circuit with resistance $R$, capacitance $C = \frac{\rho}{cR}$, and inductance $L = \frac{\rho R}{c}$, where $c$ is the speed of light in a vacuum. We assume that the antenna feed is connected to a time-varying inductor with inductance $L(t)$. Given a static inductance value $L_0$ and modulating inductance function $h(t)$, we assume that $L(t)$ has general form
\begin{equation}\label{eq:general-inductance}
    L(t) = L_0\,(1+h(t)).
\end{equation}
Fig. \ref{fig:ltv_chu} shows the equivalent circuit model for this antenna.

The $Q$-factor of a resonant antenna is defined as the ratio of the stored energy to the radiated energy per cycle. Letting $W_{\textrm{stored}}(t)$ denote the instantaneous stored energy and $P_{\textrm{total}}(t)$ denote the instantaneous radiated power, we define the time-varying quality factor $Q(t)$ of the circuit at angular frequency $\omega$ as
\begin{equation}\label{eq:Q1}
Q(t) = \frac{\omega \cdot W_{\textrm{stored}}(t)}{P_{\textrm{total}}(t)}.
\end{equation} 

\noindent We can decompose the instantaneous $Q$-factor in terms of the static Chu $Q$-factor lower bound $Q_{\textrm{Chu}}$ and the time-dependent factor $\Delta Q(t)$, giving
\begin{equation}
    Q(t) = Q_{\textrm{Chu}} + \Delta Q(t).
\end{equation}
Prior work shows evidence that the $Q$-factor can be increased and decreased through different modulation strategies \cite{LiEtAlBeyondChusLimitWith2019, FrittsEtAlSpaceTimeModulationMultimode2025}.

At each time $t$, the instantaneous $Q$-factor differs from the Chu limit by $\Delta Q(t)$. Understanding the extent to which this excess can be exploited requires a careful examination of two key parameters: the modulation function $h(t)$ and the modulation period $T_0 = t_2 - t_1$. For simplicity we assume that $t_1 = 0$ in the remainder of the paper. The excess $\Delta Q(t)$ is inherently bounded by the rate at which the system can reconfigure its reactive energy storage. A slower modulation allows the inductor to remain near its maximum value for longer durations, thereby sustaining a larger effective $Q$ over a full cycle. Conversely, faster modulation may limit the peak achievable excess due to the finite response time of the system. Moreover, the modulation period $T_0$ determines the temporal window over which the Chu limit can be violated, suggesting a fundamental trade-off between the magnitude of $\Delta Q$ and the modulation time. To formulate this problem mathematically, we seek the pair \big($T_0,\alpha$\big) that satisfy
\begin{equation}\label{eq:condition}
        \int_0^{T_0} \Delta Q(t)\,\textrm{d}t \leq -\alpha \,Q_{\textrm{Chu}},
\end{equation}
where $\alpha >0$ represents the $Q$-factor excess as a percentage of the Chu's $Q$-factor. Our goal is then to find the smallest modulation time $T_0$ satisfying \eqref{eq:condition}. In the following, we discuss different modulation strategies and their effect on the instantaneous $Q$-factor.

\section{Limitations of Sinusoidal Modulation}

\begin{figure}
\begin{center}
\begin{circuitikz}[scale=0.65] 
\draw
(0,0) to [short, *-] (6,0)
to [R, l_=$R$] (6,3)
(0,3) to [short, *-, i=$\hspace{-0.5cm}I(t)$] (1.5,3) to [variable cute inductor, l = $L(t)$] (2,3) to [C, l = $C$] (5,3) to (6,3)
(5,0) to [L, -*, l =$L$] (5,3)
(0,3) to [american voltage source, l = $V(t)$] (0,0);
\end{circuitikz}
\end{center}
\caption{\label{fig:ltv_chu} Circuit diagram of a Chu antenna connected to a time-varying inductor.}
\vspace{-0.2cm}
\end{figure}
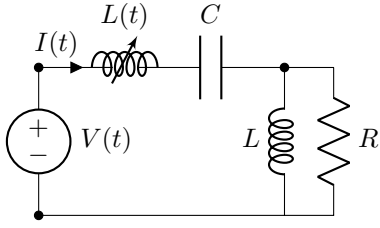

We first assume a sinusoidal modulation function
\begin{equation}\label{eq:Lt}
    L(t) = L_0\,(1+\textrm{cos}(\omega_M\,t)),
\end{equation}
where $\omega_M$ is the modulation frequency. The instantaneous $Q$-factor $Q(t)$ is then given by
\begin{equation}\label{eq:inst-Q-inductor}
    Q(t) = Q_{\textrm{Chu}} + \underbrace{\frac{\omega\,L(t)}{R}}_{=~\Delta Q(t)}.
\end{equation}
After injecting the expression of $L(t)$ from (\ref{eq:Lt}) in (\ref{eq:inst-Q-inductor}), the inequality in (\ref{eq:condition}) becomes
\begin{equation}\label{eq:condition2}
    \omega_M\,T_0 + \textrm{sin}(\omega_M\,T_0) \leq -\omega_M\, \alpha\,Q_{\textrm{Chu}}.
\end{equation}
Using this inequality, we can determine the minimum $T_0$ required to exceed the Chu limit by a factor $\alpha$. 

Fig. \ref{fig:Qchu-excess} shows how the excess fraction $\alpha$ varies as a function of the modulation time $T_0$ at frequency $f = 1$ GHz. It is seen how achieving a very small excess above the Chu limit requires a minimum modulation time $T_0$ on the order of $1/f$. For a small yet not negligible excess fraction, $T_0$ must approach microseconds to seconds. This non-linear relationship between $\alpha$ and $T_0$ suggests a fundamental trade-off between the achievable $Q$-factor excess and the temporal window over which the system must maintain the time-varying condition. In other words, exceeding the Chu's $Q$-factor limit comes at the cost of modulation over a sufficiently long interval $T_0$.\vspace{-.3cm}

\begin{figure}[h!]
    \centering
    \includegraphics[scale=0.23]{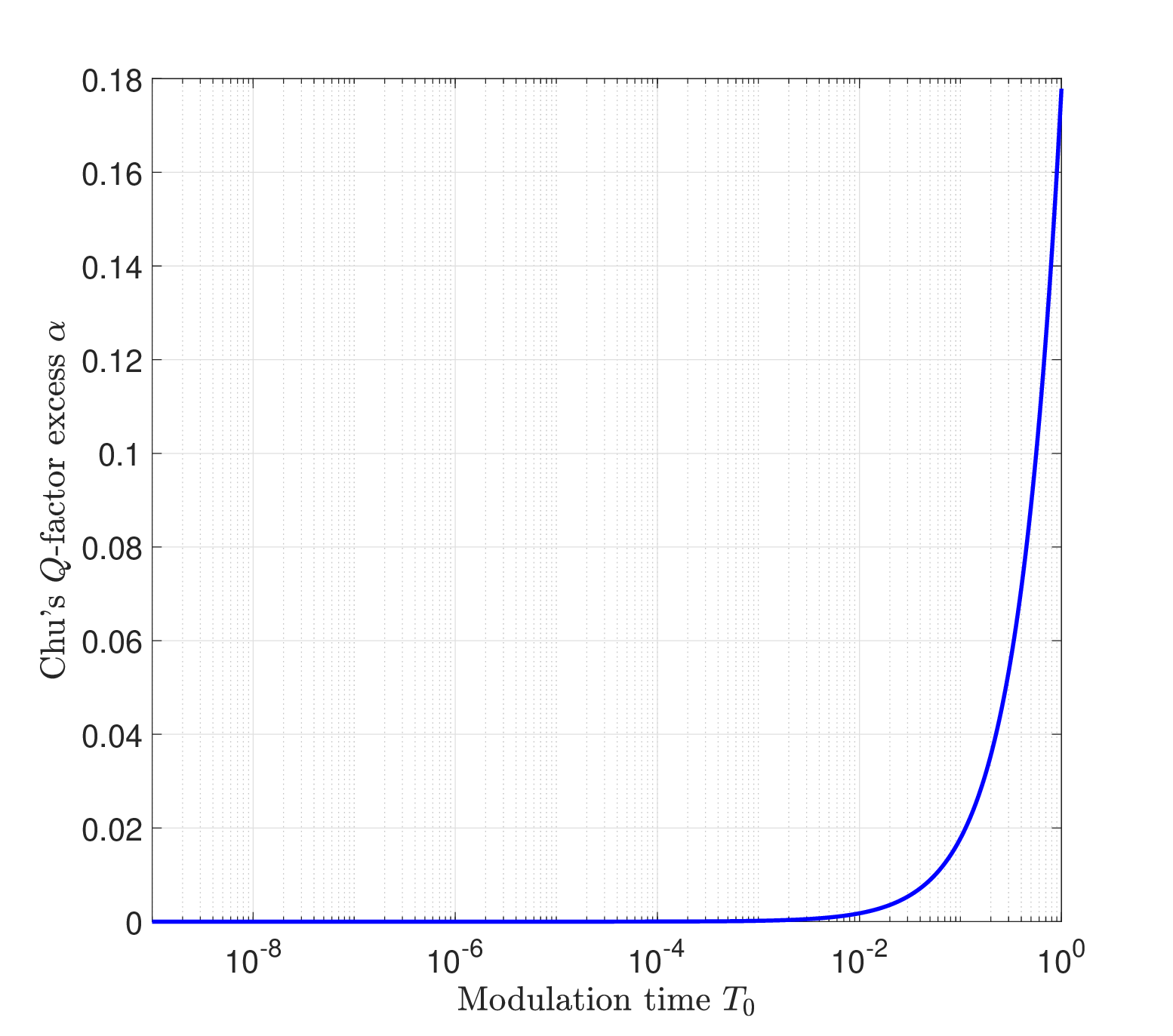}
    \caption{Modulation time $T_0$ required to achieve a given $Q$-factor excess $\alpha$ for a time-varying inductor with $L(t) = L_0(1 + \cos(\omega_M t))$ with $f=1\textrm{GHz}$, $\omega_M = \omega$, and $\rho = \lambda/10$.}
    \label{fig:Qchu-excess}
\end{figure}

The sinusoidal modulation analyzed until now may not be optimal for maximizing $Q$-factor excess factor $\alpha$. This is because the oscillatory term $\sin(\omega_M T_0)$ creates favorable and unfavorable modulation periods. This raises a natural question: can we design a time-varying inductance $L(t)$ to achieve larger excess than the cosine variation? 

\section{Optimal Inductance Modulation Condition}

In this section, we derive optimal time-varying inductance $L(t)$ that maximizes the Chu limit excess. This approach systematically explores the fundamental limits of time-varying reactance for surpassing the Chu limit.

The time-varying inductor affects both the antenna's radiated power and the stored energy. Let $L_a$ be the equivalent antenna inductance and $R_a$ be the radiation resistance such that $Q_\textrm{Chu} = \frac{\omega L_a}{R_a}$.
The total reactive energy $W_{\textrm{stored}}(t)$ of this time-modulated system is the sum of two distinct energy reservoirs: the antenna near-field energy $W_a$ and the modulating inductor energy $W_L(t)$. For a current $I(t)$, we have
\begin{equation}\label{eq:formula-Wstored}
        W_{\textrm{stored}}(t) = W_a(t) + W_L(t)
        = \frac{1}{2} L_a |I(t)|^2 + \frac{1}{2}L(t)\,|I(t)|^2.
\end{equation}
Substituting (\ref{eq:general-inductance}) back into (\ref{eq:formula-Wstored}) gives
\begin{equation}\label{eq:Wstored}
    W_{\textrm{stored}}(t) = \frac{1}{2} |I(t)|^2 \Big(L_a + L_0(1 + h(t))\Big).
\end{equation}
Similarly, the total power $P_{\textrm{total}}(t)$ leaving the system is the sum of the radiated power $P_{\textrm{rad}}(t)$, the internal ohmic losses of the modulator $P_{\textrm{loss}}(t)$, and the power $P_{\textrm{mod}}(t)$ injected by an external modulator to perform work on the magnetic field. We then have
\begin{equation}\label{eq:tot-power1}
    \begin{aligned}
        P_{\textrm{total}}(t) &= P_{\textrm{rad}}(t) + P_{\textrm{loss}}(t) + P_{\textrm{mod}}(t)\\
        &= \frac{1}{2}R_a |I(t)|^2 + \frac{1}{2}R_m |I(t)|^2 + P_{\textrm{mod}}(t),
    \end{aligned}
\end{equation}
where $R_m$ represents the modulator's internal loss resistance. 

To determine the modulation power $P_{\textrm{mod}}(t)$, we examine the instantaneous voltage across the time-varying inductor, derived from the rate of change of the magnetic flux $\Phi = L(t)I(t)$ as
\begin{equation}
V_L(t) = \frac{\textrm{d}\Phi}{\textrm{d}t} = L(t)\frac{\textrm{d}I(t)}{\textrm{d}t} + I(t)\frac{\textrm{d}L}{\textrm{d}t}.
\end{equation}
The instantaneous power $P_{\textrm{mod}}(t)$ is then
\begin{equation}
    \begin{aligned}
    P_{\textrm{mod}}(t) &= V_L(t)\,I(t)= L(t)\,I(t)\,\frac{\textrm{d}I}{\textrm{d}t} + I^2(t)\,\frac{\textrm{d}L}{\textrm{d}t}.
    \end{aligned}
\end{equation}
While the first term $L(t)\,i(t)\,\frac{\textrm{d}I(t)}{\textrm{d}t}$ represents the standard reactive power (which averages to zero over a cycle), the second term $I^2(t)\frac{\textrm{d}L}{\textrm{d}t}$ represents the real power exchange required to modify the energy state of the inductor. As a result, the total power in (\ref{eq:tot-power1}) becomes
\begin{equation}
    P_{\textrm{total}}(t) = \frac{1}{2}R_a |I(t)|^2 + \frac{1}{2}R_m |I(t)|^2 + \frac{1}{2}\frac{\textrm{d}L}{\textrm{d}t} |I(t)|^2.
\end{equation}
Substituting the derivative of the inductance $L(t)$ in (\ref{eq:general-inductance}) yields
\begin{equation}\label{eq:Ptotal}
P_{\textrm{total}}(t) = \frac{1}{2} |I(t)|^2 \cdot \Big( R_a + R_m + L_0\, h'(t) \Big).
\end{equation}

\noindent Finally, by substituting the expressions of $W_{\textrm{stored}}(t)$ in (\ref{eq:Wstored}) and $P_{\textrm{total}}(t)$ in (\ref{eq:Ptotal}) back into (\ref{eq:Q1}), we obtain
\begin{equation}\label{eq:Q-general-f}
Q(t, h) = \frac{\omega \big( L_a + L_0(1 + h(t)) \big)}{R_a + R_m + L_0 \,h'(t)} \triangleq \frac{N(h)}{D(h')},
\end{equation}
where we make the dependence of $Q$ on $h(t)$ explicit. This expression shows how the derivative $h'(t)$ acts as an additional resistance, effectively penalizing rapid changes in the impedance state. In the rest of this paper, we will call this the modulation resistance $R_{\textrm{mod}} = L_0 \,h'(t)$.

\subsection{Performance bounds}

The $Q$-factor can be bounded under the condition that the modulation function is Lipschitz continuous, as shown in the following lemma.

\begin{lemma}\label{lemma:Q-bound} Let the modulation function $h(t)$ be Lipschitz continuous with derivative bounded as $R_{\textrm{lwr}} \leq L_0 h'(t) \leq R_{\textrm{upp}}$. In addition, assume that modulation function does not result in negative instananeous power. Then, the average $Q$-factor can be bounded as
\begin{equation}
    0 \leq \frac{\omega(L_a + L_0(1 +h_0))}{R_a + R_m - R_{\textrm{lwr}}} \leq Q \leq \frac{\omega(L_a + L_0(1 +h_0))}{R_a + R_m + R_{\textrm{upp}}},
\end{equation}
where $h_0 = \frac{1}{T_0}\int_0^{T_0} h(t) \, \textrm{d}t$.
\end{lemma}

\begin{proof}
The proof follows immediately by substituting the bounds on the modulation resistance into the average $Q$-factor $Q = \frac{1}{t_2 - t_1} \int_{t_a}^{t_2} Q(t)\, \textrm{d}t$ for $t_1=0$ and $t_2=T_0$.
\end{proof}

Lemma \ref{lemma:Q-bound} demonstrates that changes to the $Q$-factor from the time-varying inductor are inherently limited by the maxima and minima of the modulation resistance. These bounds are generally loose, however, since large values of the modulation resistance will concurrently affect the modulation inductance $L_0(1 + h(t))$. This effect will lead to overall small changes in the $Q$-factor, as shown in the case of sinusoidal modulation. We now derive the form of the optimal modulation function.

\subsection{Variational optimization}

We define the functional $\mathcal{J}[h]$ to be maximized over a modulation period $T_0$ as
\begin{equation}\label{eq:condition-violation}
    \mathcal{J}[h] = \int_0^{T_0} \big(Q(t,h) - Q_{\textrm{Chu}}\big) \,\textrm{d}t \leq -\alpha \,Q_{\textrm{Chu}}.
\end{equation}

\noindent To find the optimal $h(t)$, we use the calculus of variations by introducing a perturbation $\eta(t)$ such that $\eta(0) = \eta(T_0) = 0$. Let $\widehat{h}(t) = h(t) + \epsilon\, \eta(t)$. The extremum condition $\frac{\textrm{d}\mathcal{J}}{d\epsilon} \big|_{\epsilon=0} = 0$ yields
\begin{equation}\label{eq:J=0}
    \int_{0}^{T_0} \left( \frac{\partial \mathcal{L}}{\partial h} \,\eta(t) + \frac{\partial \mathcal{L}}{\partial h'} \,\eta'(t) \right) \textrm{d}t = 0.
\end{equation}

\noindent Applying integration by parts to the second term gives
\begin{equation}\label{eq:integration-by-part}
    \begin{aligned}
        \int_{0}^{T_0} \frac{\partial \mathcal{L}}{\partial h'} \frac{\textrm{d}\eta}{\textrm{d}t} \textrm{d}t &= \left[ \frac{\partial \mathcal{L}}{\partial h'} \eta(t) \right]_0^{T_0} - \int_{0}^{T_0} \eta(t) \frac{\textrm{d}}{\textrm{d}t} \left( \frac{\partial \mathcal{L}}{\partial h'} \right) \textrm{d}t\\
        &= - \int_{0}^{T_0} \eta(t) \frac{\textrm{d}}{\textrm{d}t} \left( \frac{\partial \mathcal{L}}{\partial h'} \right) \textrm{d}t,
    \end{aligned}  
\end{equation}
where the last equality follows from the fact $\eta(0)=\eta(T_0)=0$. We then inject (\ref{eq:integration-by-part}) into (\ref{eq:J=0}) and factor out the common perturbation $\eta(t)$ to rewrite
\begin{equation}
\int_{0}^{T_0} \eta(t) \left[ \frac{\partial \mathcal{L}}{\partial h} - \frac{\textrm{d}}{\textrm{d}t} \left( \frac{\partial \mathcal{L}}{\partial h'} \right) \right] \textrm{d}t = 0.
\end{equation}
By the fundamental lemma of the calculus of variations, if this integral vanishes for any arbitrary smooth function $\eta(t)$, then the integrand itself must be zero for all $t \in [0, T_0]$. This directly yields the so-called Euler-Lagrange condition
\begin{equation}\label{eq:Eurler-Lagrange}
\frac{\partial \mathcal{L}}{\partial h} - \frac{\textrm{d}}{\textrm{d}t} \left( \frac{\partial \mathcal{L}}{\partial h'} \right) = 0.
\end{equation}

\noindent To evaluate (\ref{eq:Eurler-Lagrange}) for our context, we use the definition of the instantaneous $Q$-factor $Q(t,h)$ in (\ref{eq:Q-general-f}) to get
\begin{subequations}\label{eq:derivatives}
    \begin{align}
        \frac{\partial \mathcal{L}}{\partial h} &= \frac{\omega L_0}{D(h')},\quad\frac{\partial \mathcal{L}}{\partial h'} = - \frac{N(h) L_0}{D(h')^2},\\
        \frac{\textrm{d}}{\textrm{d}t} \left( \frac{\partial \mathcal{L}}{\partial h'} \right) &= -L_0 \left( \frac{D(h') \omega L_0 h' - 2N(h) L_0 h''}{D(h')^3} \right).
    \end{align}
\end{subequations}

\noindent After substituting these terms into (\ref{eq:Eurler-Lagrange}), we arrive at the governing non-linear second-order differential equation:
\begin{equation}\label{eq:governing_ODE_1}
    \omega L_0 D(h')^2 + \omega L_0^2 D(h') h' - 2 \,N(h) L_0^2 h'' = 0.
\end{equation}
Substituting $N(h)$ and $D(h')$ from (\ref{eq:Q-general-f}) yield
\begin{equation}\label{eq:governing_ODE}
    \begin{aligned}
    &2(h')^2 + \frac{3 (R_a+R_m)}{L_0} h' \\
    &\hspace{1cm}+ \left( \frac{R_a+R_m}{L_0} \right)^2 = 2 \left( \frac{L_a}{L_0} + 1 + h \right) h''.    
    \end{aligned}
\end{equation}
This equation, which we term the inductive modulation condition, defines the necessary trajectory for $h(t)$ to optimize the quality factor excess relative to the Chu limit. It reveals the relationship between the modulation velocity $h'$ and the acceleration $h''$. It also indicates that the optimal path for switching is independent of the modulation period $T_0$. This independence suggests that the switching dynamics are a fundamental property of the circuit's constants (the ratios of $R_a$, $R_m$, and $L_0$), allowing the effect of the switching profile to scale across varying symbol rates.

\section{Optimality of Bang-Bang switching for inductive antenna modulation}

The governing differential equation \eqref{eq:governing_ODE} evaluates the efficiency of time-modulated reactances. Recall from (\ref{eq:Q-general-f}) that the total effective resistance of the modulated antenna is
\begin{equation}\label{eq:denum}
    \begin{aligned}
    R_{\textrm{eff}}(t) = R_a + R_m + R_{\textrm{mod}}(t)
    = R_a + R_m + L_0 h'(t).
    \end{aligned}
\end{equation}

\noindent For any modulation function $h(t) \in C^1$, the non-zero derivative $h'(t)$ induces a persistent activation of the modulation resistance $R_{\textrm{mod}}$. This leads to an increase in the total effective dissipation, $R_{\textrm{eff}}(t)$, which monotonically suppresses the instantaneous quality factor. Mathematically, the gains in the energy-storage numerator $N(h)$ are effectively neutralized by the concurrent rise in the power-dissipating denominator $D(h')$. This inherent coupling between the rate of change of the reactance and the associated modulation resistance establishes a fundamental efficiency ceiling. Any differentiable trajectory $h(t) \in C^1$ in (\ref{eq:governing_ODE}) incurs a non-zero modulated resistance within the denominator $D(h)$ in (\ref{eq:denum}). The following lemma confirms that smooth modulation functions inherently limit the instantaneous quality factor.

\begin{lemma} The functional $\mathcal{J}$ in (\ref{eq:condition-violation}) is maximized when the modulation function $h(t)$ belongs to the class of piecewise-constant functions $\Psi$. 
\end{lemma}

\begin{proof}
Consider an interval where the inductor is in a steady state $h(t) = h_{\textrm{max}}$. In this regime, $h' = 0$ and $h'' = 0$, reducing the denominator to its physical minimum $R_a + R_m$. Under these conditions, the system achieves the maximum possible $Q$-factor. Transitioning between these steady states requires a non-zero $h'$. However, the energetic cost of this transition can be mathematically nullified if the transition velocity $f'$ occurs at the precise instants where $I(t) = 0$. By concentrating the modulation into instantaneous switches, the system bypasses the modulation resistance. The optimal control strategy then collapses into a Bang-Bang switching profile.
\end{proof}

\noindent It is worth highlighting that when using a switching function $h(t)$ between two discrete extrema, the modulation cycle undergoes a functional bifurcation:
\begin{itemize}[leftmargin=*]
    \item During the \textit{static states} where $h(t)$ remains at its maximum or minimum value, the modulation velocity $h'$ vanishes. Consequently, the modulation resistance $R_{\textrm{mod}}$ is zero, and hence maximizing the $Q$-factor.

    \item During the \textit{zero-crossing transition state}, the derivative $f'$ theoretically approaches a Dirac delta impulse. The overall energy remains zero if the state transition is synchronized when the current $I(t)=0$.
\end{itemize}

\noindent As a result, the Bang-Bang solution emerges not as a heuristic switching choice but as a mathematical necessity. It represents the only modulation class capable of maintaining $R_{\textrm{mod}} = 0$ for the vast majority of the duty cycle, thereby bypassing the inherent efficiency of smooth modulation functions. This facilitates pushing the $Q$-factor beyond the Chu limit.

\section{Bounding the switching time limit for Antenna modulation}

Every antenna acts as a resonator that stores energy in its surrounding electric and magnetic fields \cite{pozar2011microwave}. When the drive signal is removed, this energy does not vanish instantly. It follows an exponential decay governed by the antenna's quality factor $Q$ and the carrier frequency $f_c$. The field amplitude $A(t)$ is expressed as
\begin{equation}\label{eq:field-resonator}
    A(t) = A_0 \,e^{-t/\tau_{\text{ant}}}
\end{equation}
where the antenna time constant is defined by
\begin{equation}\label{eq:tau-ant}
    \tau_{\text{ant}} = \frac{Q}{\pi f_c}.
\end{equation}

\noindent The antenna time constant $\tau_{\text{ant}}$ serves as an indicator for the ``memory'' inherent to the radiator's reactive near-fields. In the absence of state-switching, an electrically small antenna (ESA) behaves as a high-$Q$ resonator that continues to oscillate well beyond the ending time of the driving signal. This persistent natural response is fundamentally bounded by the Chu-Harrington limit on energy storage and creates a temporal overlap between successive symbols. The resulting inter-symbol interference (ISI) imposes a rigid constraint on the achievable error-free data rate in conventional LTI systems. For this reason, it is important to use DAM to actively reset the field states within the switching period.

To send data reliably, we must decide how unexcited the antenna must be before the excitation starts. We define the isolation requirement $S$ in dB as the target signal-to-interference ratio. Mathematically, we solve for the time $t_{\text{decay}}$ it takes for the field in (\ref{eq:field-resonator}) to drop below a threshold $S$. That is to say,
\begin{equation}\label{eq:tdelay}
    e^{-t_{\text{decay}}/\tau_{\text{ant}}} \leq 10^{-S/20} \iff t_{\text{decay}} \geq \tau_{\text{ant}} \,\ln\big(10^{S/20}\big).
\end{equation}
Injecting back the expression of $\tau_{\text{ant}}$ in (\ref{eq:tau-ant}) back into (\ref{eq:tdelay}) yields the expression of the smallest time delay $t_{\text{decay}}$:
\begin{equation}\label{eq:tdelay2}
    t_{\text{decay}} = \frac{Q}{\pi f_c} \ln(10^{S/20}).
\end{equation}

For antenna modulation, we use a switch with a physical transition time $\tau_{\text{sw}}$ to reset the antenna state faster than it would happen naturally. However, we cannot spend the entire symbol period switching. For this reason, we introduce a fractional occupancy parameter $\beta$ to model the maximum allowable fraction of a symbol duration we can use for switch transition. For antenna modulation to be effective, we require that the time allocated for switch transition $\tau_{\text{sw}}$  must be faster than a fraction of the antenna's natural memory \vspace{-0.1cm}
\begin{equation}\label{eq:tswitch}
    \tau_{\text{sw}} < \beta \cdot t_{\text{decay}}.
\end{equation}
Injecting (\ref{eq:tdelay2}) back into (\ref{eq:tswitch}) yields the following upper-bound on the switching time\vspace{-0.1cm}
\begin{equation}
    \tau_{\text{sw}} < \frac{\beta \cdot Q \cdot \ln(10^{S/20})}{\pi f_c}.
\end{equation}
Substituting the $Q$-factor with its Chu limit upper-bound $Q=1/(k \rho)^3$ gives\vspace{-0.1cm}
\begin{equation}
    \tau_{\text{sw}} < \frac{\beta \ln\big(10^{S/20}\big)}{\pi f_c (k \rho)^3}.
\end{equation}
This bound reveals a counterintuitive observation: the smaller the antenna size $\rho$ is, the higher the bound on the switching time. This means that the hardware requirement for the switch is easier to meet for compact antennas because the switch has a larger window of time to truncate the natural decay.

Finally, it is possible to involve a desired level of bit error rate  (BER) by connecting it to the isolation requirement $S$. For instance, for BPSK modulation, $S$ is linked to $P_e$ through the inverse of the error function $Q_{\textrm{inv}}(\cdot)$ as follows \cite{proakis2001digital}
\begin{equation}
    S = 10 \log_{10} \left[ \frac{1}{2} \left( Q_{\text{inv}}^{-1}(P_e) \right)^2 \right],
\end{equation}
from to which we obtain the switching time as a function of the bit error rate as follows:
\begin{equation}
    \tau_{\text{sw}} < \frac{\beta}{2\pi f_c (k\rho)^3} \ln \left[ \frac{1}{2} \left( Q_{\text{inv}}^{-1}(P_e) \right)^2 \right]
\end{equation}

\noindent Note that this bound contains the term $\ln\left[\frac{1}{2}\left(Q_{\text{inv}}(P_e)\right)^2\right]$ which is positive only when $Q_{\text{inv}}(P_e) > \sqrt{2}$, or equivalently when $P_e < Q(\sqrt{2}) \approx 0.0786$. Above that BER, the isolation requirement $S$ means DAM offers little advantage.

Fig.~\ref{fig:heatmap-Qchu-excess} shows the fundamental trade-off
between antenna miniaturization, the decoder reliability for BPSK modulation, and the switching speed for DAM. Smaller antennas with gives the switch more time to reset the stored energy before the next symbol. For example, at a BER of $10^{-5}$, an antenna with
$\rho/\lambda = 0.02$ requires only a switch faster than roughly
$10$\,ns. This is achievable with commercial GaN FETs. However, an antenna with $\rho/\lambda = 0.08$ demands sub-nanosecond transitions. The red dashed
line marks the BER threshold ($P_e \approx 0.0786$) above which the
bound becomes infeasible.\vspace{-0.4cm}

\begin{figure}[h!]
    \centering
    \includegraphics[scale=0.3]{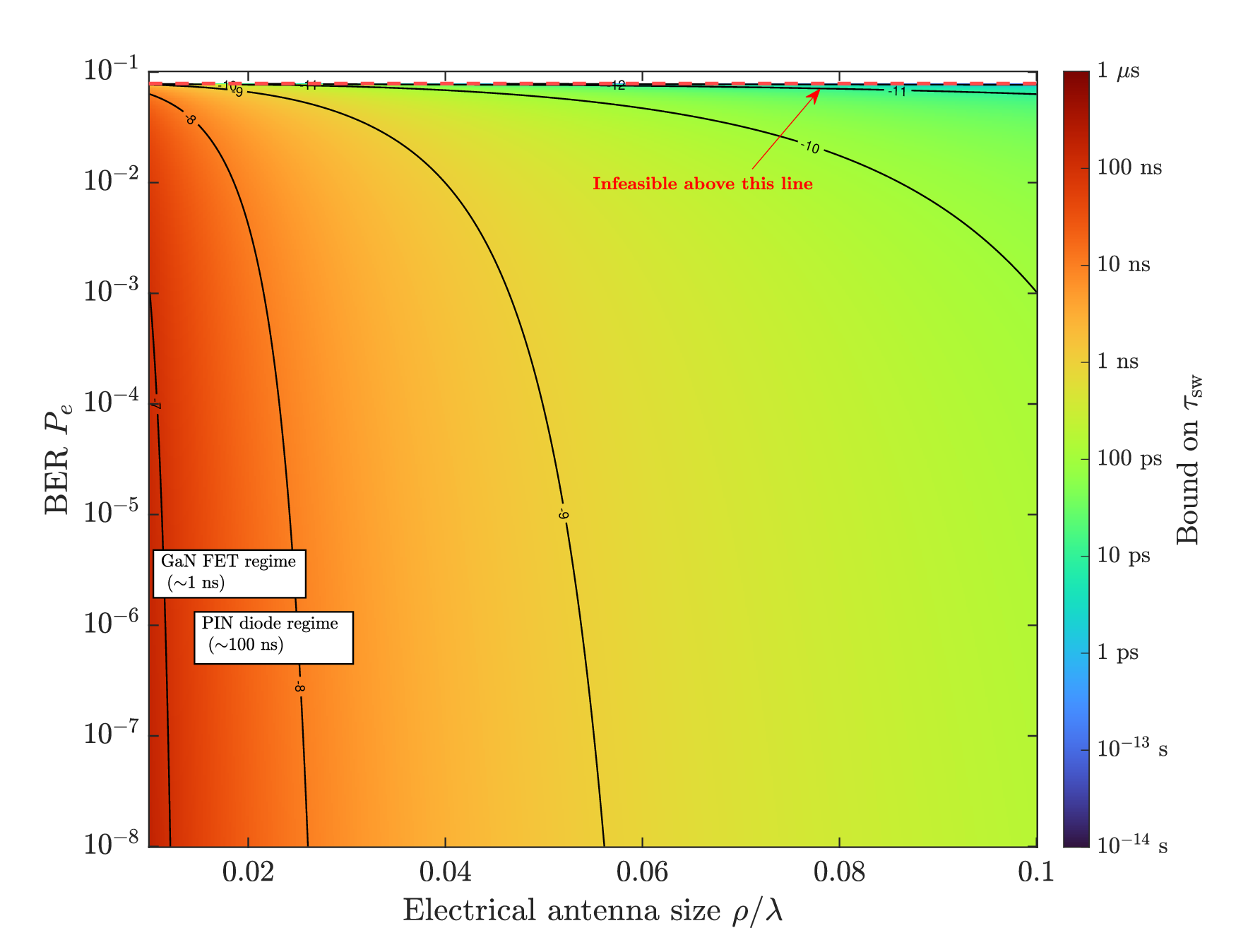}
    \caption{Heatmap of the maximum allowable switch transition time
    $\tau_{\mathrm{sw}}^{\max}$ as a function of electrical antenna size
    $\rho/\lambda$ and bit error rate $P_e$, with $f_c = 1$\,GHz, and $\beta = 0.25$. The red dashed line marks the critical BER threshold
    $P_e \approx 0.0786$, above which the bound is infeasible.}
    \label{fig:heatmap-Qchu-excess}
    \vspace{-0.1cm}
\end{figure}

\vspace{0.8cm}
\section{Conclusion}
We have shown that a Bang-Bang switching profile is the analytical necessity for surpassing the Chu limit, as it bypasses the efficiency penalties of smooth modulations. The derived differential  equation confirms that the optimal switching trajectory is fundamental to the system's constants and scale across symbol rates. Finally, we also established a bound on the switching time showing that antenna miniaturization widens the time window for direct antenna modulation.

\bibliographystyle{IEEEtran}
\bibliography{esa_mod_refs_rc,references}

@ARTICLE{BroadbooksEtAlNovelContinuousDirectAntenna2025,
	author = {Broadbooks, Ian M. and Smith, Matthew C. and Radhakrishnan, Rohith M. and Sievenpiper, Daniel F.},
	journal = {IEEE Trans. Antennas Propag.},
	title = {A Novel Continuous Direct Antenna Modulation System Through Varactor Diode Tuning},
	year = {2025},
	volume = {73},
	number = {4},
	pages = {2416-2426},
	doi = {10.1109/TAP.2025.3540276},
	ISSN = {1558-2221},
	month = {Apr.}
}

@ARTICLE{SchabEtAlEnergySynchronousDirectAntenna2020,
	author = {Schab, Kurt and Huang, Danyang and Adams, Jacob J.},
	journal = {IEEE Open J. Antennas Propag.},
	title = {An Energy-Synchronous Direct Antenna Modulation Method for Phase Shift Keying},
	year = {2020},
	volume = {1},
	number = {},
	pages = {41-46},
	doi = {10.1109/OJAP.2020.2972842},
	ISSN = {2637-6431},
	month = {}
}

@ARTICLE{ChurchEtAlUhfElectricallySmallBox2014,
	author = {Church, Justin and Chieh, Jia-Chi Samuel and Xu, Lu and Rockway, John D. and Arceo, Diana},
	journal = {IEEE Antennas Wireless Propag. Lett.},
	title = {{UHF} Electrically Small Box Cage Loop Antenna With an Embedded Non-{Foster} Load},
	year = {2014},
	volume = {13},
	number = {},
	pages = {1329-1332},
	doi = {10.1109/LAWP.2014.2337112},
	ISSN = {1548-5757},
	month = {Jul.}
}

@ARTICLE{LoghmanniaManteghiBroadbandParametricImpedanceMatching2022,
	author = {Loghmannia, Pedram and Manteghi, Majid},
	journal = {IEEE Antennas Propag. Mag.},
	title = {Broadband Parametric Impedance Matching for Small Antennas Using the {Bode-Fano} Limit: Improving on {Chu}’s limit for loaded small antennas},
	year = {2022},
	volume = {64},
	number = {5},
	pages = {55-68},
	doi = {10.1109/MAP.2021.3089997},
	ISSN = {1558-4143},
	month = {Oct}
}

@article{JayathurathnageEtAlTimeVaryingComponentsEnhancing2021,
	title = {Time-Varying Components for Enhancing Wireless Transfer of Power and Information},
	author = {Jayathurathnage, Prasad and Liu, Fu and Mirmoosa, Mohammad S. and Wang, Xuchen and Fleury, Romain and Tretyakov, Sergei A.},
	journal = {Phys. Rev. Appl.},
	volume = {16},
	issue = {1},
	pages = {014017},
	numpages = {15},
	year = {2021},
	month = {Jul}
}

@article{MekawyEtAlParametricEnhancementRadiationFrom2021,
	title = {Parametric Enhancement of Radiation from Electrically Small Antennas},
	author = {Mekawy, Ahmed and Li, Huanan and Radi, Younes and Al\`u, Andrea},
	journal = {Phys. Rev. Appl.},
	volume = {15},
	issue = {5},
	pages = {054063},
	numpages = {11},
	year = {2021},
	month = {May}
}

@ARTICLE{FrittsEtAlSpaceTimeModulationMultimode2025,
	author = {Fritts, Zachary and Babaee, Amirhossein and Young, Steve M. and Grbic, Anthony},
	journal = {IEEE Trans. Antennas Propag.},
	title = {Space–Time Modulation of a Multimode Electrically Small Antenna for Increased Matching and Efficiency Bandwidths},
	year = {2025},
	volume = {73},
	number = {3},
	pages = {1308-1320},
	doi = {10.1109/TAP.2024.3508081},
	ISSN = {1558-2221},
	month = {Mar.}
}

@article{LiEtAlBeyondChusLimitWith2019,
	title = {Beyond {Chu's} Limit with {Floquet} Impedance Matching},
	author = {Li, Huanan and Mekawy, Ahmed and Al\`{u}, Andrea},
	journal = {Phys. Rev. Lett.},
	volume = {123},
	issue = {16},
	pages = {164102},
	numpages = {6},
	year = {2019},
	month = {Oct},
}

@article{chu1948physical,
  title={Physical limitations of omni-directional antennas},
  author={Chu, Lan Jen},
  journal={J. Appl. Phys.},
  volume={19},
  number={12},
  pages={1163--1175},
  year={1948},
  publisher={American Institute of Physics}
}

@book{proakis2001digital,
  title={Digital communications},
  author={Proakis, John G and Salehi, Masoud},
  volume={4},
  year={2001},
  publisher={McGraw-hill New York}
}

@book{pozar2011microwave,
  title={Microwave engineering},
  author={Pozar, David M},
  year={2011},
  publisher={John wiley \& sons}
}

@article{skrivervik2002pcs,
  title={{PCS} antenna design: The challenge of miniaturization},
  author={Skrivervik, A Ks and Zurcher, J-F and Staub, O and Mosig, JR},
  journal={IEEE Antennas Propag. Mag},
  volume={43},
  number={4},
  pages={12--27},
  year={2002},
  publisher={IEEE}
}

@article{mostafa2023antenna,
  title={Antenna bandwidth engineering through time-varying resistance},
  author={Mostafa, MH and Ha-Van, N and Jayathurathnage, P and Wang, Xuchen and Ptitcyn, G and Tretyakov, SA},
  journal={Applied Physics Letters},
  volume={122},
  number={17},
  year={2023},
  publisher={AIP Publishing}
}
\end{document}